\documentclass[onecolumn,amsmath,amssymb,10pt,aps]{revtex4}
  
\usepackage[utf8]{inputenc}
\usepackage[T1]{fontenc}

\usepackage{amsmath}
\usepackage{amsfonts}
\usepackage{graphicx}
\usepackage{yfonts}
\usepackage{color}
\usepackage[normalem]{ulem}
\usepackage{amsthm}
\usepackage{bm}
\usepackage{bbm}
\usepackage{mathtools}
\usepackage{array}
\usepackage{placeins}
\usepackage{enumitem}

\def\<{\langle}
\def\>{\rangle}
\newcommand{\tr}{\mathrm{Tr}}
\newcommand{\Tr}{\mathrm{Tr}}

\DeclareMathAlphabet\mathbfcal{OMS}{cmsy}{b}{n}

\mathchardef\mhyphen="2D 

\newtheorem{Theorem}{Theorem}
\newtheorem{Lemma}{Lemma}
\newtheorem{Corollary}{Corollary}
\newtheorem{Definition}{Definition}
\newtheorem{Remark}{Remark}

\begin{document}

\title{Measures from conical 2-designs depend only on two constants}

\author{Katarzyna Siudzi\'{n}ska}
\affiliation{Institute of Physics, Faculty of Physics, Astronomy and Informatics, Nicolaus Copernicus University in Toru\'{n}, ul.~Grudzi\k{a}dzka 5, 87--100 Toru\'{n}, Poland}

\begin{abstract}
Quantum measurements are important tools in quantum information, represented by positive, operator-valued measures. A wide class of symmetric measurements is given via generalized equiangular measurements that form conical 2-designs. We show that only two positive constants are needed to fully characterize a variety of important quantum measures constructed from such operators. Examples are given for entropic uncertainty relations, the Brukner-Zeilinger invariants, quantum coherence, quantum concurrence, and the Schmidt-number criterion for entanglement detection.
\end{abstract}

\flushbottom

\maketitle

\thispagestyle{empty}

\section{Introduction}

In quantum information theory, quantum measurements are positive, operator-valued measures (POVMs). By defitnition, a POVM $\mathcal{P}=\{P_k;\,k=1,\ldots,M\}$ consists of positive operators $P_k$ such that $\sum_{k=1}^MP_k=\mathbb{I}_d$, where $\mathbb{I}_d$ denotes the identity operator. Among the most researched classes of projective measurements, there are symmetric, informationally complete (SIC) POVMs \cite{Renes} and projectors onto the vectors from mutually unbiased bases (MUBs) \cite{Schwinger,Szarek}. Both have been generalized by Kalev and Gour to non-projective measurements: general SIC (gSIC) POVMs and mutually unbiased measurements (MUMs) \cite{Kalev,Gour}. Due to their symmetry properties, these measurements find wide theoretical and experimental applications in quantum tomography \cite{Prugovecki,Scott,ZhuEnglert,PetzRuppert,Pimenta,Bent}, state discrimination \cite{Brunner}, communication protocols \cite{Zhou,Song,Bouchard}, steering \cite{Lai2,OperationalSICs,Bene}, teleportation \cite{Siendong}, entanglement detection \cite{ESIC,EW-SIC,KalevBae,Blume,SM_Pmaps}, and many more \cite{Medendorp,Tavakoli,Tavakoli2,Smania,Hirsch}.

There have been increasing attempts to introduce alternative generalizations of SIC POVMs and MUBs \cite{semi-SIC,EOM22,EOM24,EOMq3}. $(N,M)$-POVMs are collections of $N$ mutually unbiased POVMs with $M$ elements each \cite{SIC-MUB}. If the number of elements is different, one considers generalized symmetric measurements \cite{SIC-MUB_general}. If these POVMs are taken with different weights, we recover generalized equiangular measurements (GEAMs). It has been shown that they form an informationally overcomplete set and, under additional assumptions, also conical 2-designs \cite{Graydon,Graydon2}. Despite being a relatively recent construction, GEAMs have already found applications in uncertainty relations \cite{Rastegin_GEAM}, steering inequalities \cite{steering}, Kirkwood-Dirac quasiprobabilities \cite{Kirkwood}, and entanglement detection \cite{GEAM_Pmaps}.

In this paper, we show that the knowledge of the two constants that characterize a conical 2-design is enough to calculate various quantum measures. In particular, we present this in the general cases of the Brukner-Zeilinger invariants, quantum coherence, entropic uncertainty relations, quantum concurrence, and the Schmidt-number criterion for entanglement detection. Our observations could indicate on a more general property followed by specific families of measures for all conical quantum 2-designs.

The paper is organized as follows. In Section 2, we recall the basic properties of generalized equiangular measurements, their method of construction, and conditions under which they form conical 2-designs. It is shown how to recover specific examples of measurements, like the SIC POVMs, MUBs, and their generalizations. The remaining sections contain examples of measures, constructed from GEAMs that are conical 2-designs, that are characterized only by two constants. In Section 3, we calculate entropic uncertainty relations for the Tsallis and R{\'e}nyi entropies. Section 4 deals with the Brukner-Zeilinger invariants in terms of the total variance. Quantum coherence measure using skew information is considered in Section 5. Entanglement detection methods with the Schmidt number criterion and concurrence lower bound are analyzed in Sections 6 and 7, respectively. Finally, Section 8 contains a summary of our results, together with a list of open questions for further research.

\section{Generalized equiangular measurements}

Consider a set of rank-1 projectors $P_k$ ($k=1,\ldots,M$) that act on a finite-dimensional Hilbert space $\mathcal{H}\simeq\mathbb{C}^d$. These operators form an equiangular tight frame if and only if, for all $k\neq\ell$, $\Tr(P_kP_\ell)=c\Tr(P_k)\Tr(P_\ell)$ and additionally $\sum_{k=1}^MP_k=\gamma\mathbb{I}_d$ with some positive parameter $\gamma$ \cite{Strohmer}. This notion can be generalized to positive operators of arbitrary rank \cite{GEAM}. From a collection of $N$ generalized equiangular tight frames $\mathcal{P}_\alpha=\{P_{\alpha,k}:\,k=1,\ldots,M_\alpha\}$, $\alpha=1,\ldots,N$, with elements summing up to $\sum_{k=1}^{M_\alpha}P_{\alpha,k}=\gamma_\alpha\mathbb{I}_d$, we define a positive, operator-valued measure (POVM). In addition, we impose the complementarity condition $\Tr(P_{\alpha,k}P_{\beta,\ell})=f\Tr(P_{\alpha,k})\Tr(P_{\beta,\ell})$, $\alpha\neq\beta$, between the elements from different frames \cite{PetzRuppert}.

\begin{Definition}\label{geam}(\cite{GEAM})
The generalized equiangular measurement (GEAM) $\mathcal{P}=\cup_{\alpha=1}^N\mathcal{P}_\alpha$ is a collection of $N$ generalized equiangular tight frames $\mathcal{P}_\alpha=\{P_{\alpha,k}:\,k=1,\ldots,M_\alpha\}$ such that
\begin{enumerate}[label=(\arabic*)]
\item the elements of a single frame sum up to $\sum_{k=1}^{M_\alpha}P_{\alpha,k}=\gamma_\alpha\mathbb{I}_d$ with a probability distribution $\gamma_\alpha$;
\item the total number of operators $P_{\alpha,k}$ is bounded by $\sum_{\alpha=1}^NM_\alpha=d^2+N-1$;
\item they satisfy the following trace conditions,
\begin{equation}
\begin{split}
\Tr(P_{\alpha,k})&=a_\alpha,\\
\Tr(P_{\alpha,k}^2)&=b_\alpha \Tr(P_{\alpha,k})^2,\\
\Tr(P_{\alpha,k}P_{\alpha,\ell})&=c_\alpha
\Tr(P_{\alpha,k})\Tr(P_{\alpha,\ell}),\qquad k\neq\ell,\\
\Tr(P_{\alpha,k}P_{\beta,\ell})&=
f\Tr(P_{\alpha,k})\Tr(P_{\beta,\ell}),\qquad \alpha\neq\beta,
\end{split}
\end{equation}
where the symmetry parameters
\begin{equation}
a_\alpha=\frac{d\gamma_\alpha}{M_\alpha},\qquad c_\alpha=\frac{M_\alpha-db_\alpha}{d(M_\alpha-1)},
\qquad f=\frac 1d,
\end{equation}
\begin{equation}\label{ba}
\frac 1d <b_\alpha\leq\frac 1d \min\{d,M_\alpha\}.
\end{equation}
\end{enumerate}
\end{Definition}

Note that the generalized equiangular measurements are informationally complete for $N=1$ and informationally overcomplete otherwise. In particular, inside each equiangular tight frame $\mathcal{P}_\alpha$, $M_\alpha-1$ out of the total of $M_\alpha$ operators are linearly independent.

To construct the generalized equiangular measurements, take a Hermitian orthonormal operator basis
\begin{equation}\label{Ga}
\mathcal{G}=\{G_0=\mathbb{I}_d/\sqrt{d},G_{\alpha,k}:\,k=1,\ldots,M_\alpha-1;\,\alpha=1,\ldots,N\}
\end{equation}
such that $\Tr(G_{\alpha,k})=0$. Next, introduce another family of traceless operators $H_{\alpha,k}$ via
\begin{equation}\label{H}
H_{\alpha,k}=\left\{\begin{aligned}
&G_\alpha-\sqrt{M_\alpha}(1+\sqrt{M_\alpha})G_{\alpha,k},\quad k=1,\ldots,M_\alpha-1,\\
&(1+\sqrt{M_\alpha})G_\alpha,\qquad k=M_\alpha,
\end{aligned}\right.
\end{equation}
where $G_\alpha=\sum_{k=1}^{M_\alpha-1}G_{\alpha,k}$. Finally, we find \cite{GEAM}
\begin{equation}\label{Pak}
P_{\alpha,k}=\frac{a_\alpha}{d}\mathbb{I}_d+\tau_\alpha H_{\alpha,k}
\end{equation}
with a real parameter
\begin{equation}
\tau_\alpha=\pm\sqrt{\frac{S_\alpha}{M_\alpha(\sqrt{M_\alpha}+1)^2}},\qquad
{\rm where}\qquad S_\alpha=a_\alpha^2(b_\alpha-c_\alpha),
\end{equation}
chosen in such a way that $P_{\alpha,k}$ is a positive operator. Conversely, when given a GEAM, one recovers the corresponding Hermitian operator basis via
\begin{equation}\label{Gak}
G_{\alpha,k}= \frac{1}{\tau_\alpha M_\alpha(1+\sqrt{M_\alpha})^2}\Big[\mathbb{I}_d
+\sqrt{M_\alpha}P_{\alpha,M_\alpha}-\sqrt{M_\alpha}(1+\sqrt{M_\alpha})P_{\alpha,k}\Big].
\end{equation}
Importantly, a single $\mathcal{G}$ can give rise to many inequivalent families of measurement operators \cite{SIC-MUB_general}. However, this property is not needed for further discussions of this paper.

Among the GEAMs, there is a distinguished class of measurements that form conical 2-designs \cite{Graydon,Graydon2}. By definition, $\mathcal{P}$ is a conical 2-design if
\begin{equation}\label{con}
\sum_{\alpha=1}^N\sum_{k=1}^{M_\alpha}P_{\alpha,k}\otimes P_{\alpha,k}=
\kappa_+\mathbb{I}_d\otimes\mathbb{I}_d+\kappa_-\mathbb{F}_d,
\end{equation}
where $\kappa_+\geq\kappa_->0$ and $\mathbb{F}_d=\sum_{m,n=0}^{d-1}|m\>\<n|\otimes|n\>\<m|$ denotes the flip operator. This condition is satisfied as long as $S_\alpha=a_\alpha^2(b_\alpha-c_\alpha)\equiv S$ for all $\alpha=1,\ldots,N$. The admissible range of the parameter $S$ is
\begin{equation}\label{Srange}
0<S\leq \min\left\{\frac{d\gamma_\alpha^2}{M_\alpha},\frac{d-1}{M_\alpha-1}\frac{d\gamma_\alpha^2}{M_\alpha}\right\},
\end{equation}
and then the positive coefficients $\kappa_\pm$ from eq. (\ref{con}) read
\begin{equation}\label{kappas2}
\kappa_+=\mu-\frac{S}{d},\qquad\kappa_-=S,\qquad {\rm where}\qquad \mu=\frac 1d \sum_{\alpha=1}^Na_\alpha\gamma_\alpha.
\end{equation}
Importantly, it has been shown in ref. \cite{GEAM} that there is a direct relation between the state purity $\Tr(\rho^2)$ and the index of coincidence \cite{Rastegin5}
\begin{equation}\label{pak}
\mathcal{C}(\rho)=\sum_{\alpha=1}^N\sum_{k=1}^{M_\alpha}p_{\alpha,k}^2,\qquad p_{\alpha,k}=\Tr(P_{\alpha,k}\rho)
\end{equation}
if and only if $P_{\alpha,k}$ form a conical 2-design. Then,
\begin{equation}\label{IOCN}
\mathcal{C}(\rho)=S\left(\Tr\rho^2-\frac 1d\right)+\mu,
\end{equation}
which is upper bounded by
\begin{equation}
\mathcal{C}_{\max}=\frac{d-1}{d}S+\mu.
\end{equation}
In Table 1, we present popular examples of the generalized equiangular measurements together with the explicit values of their defining parameters, including $S$ and $\mathcal{C}_{\max}$. Note that for families of quantum measurements, like the mutually unbiased bases and measurements, the corresponding constants are rescaled because they have to sum up to the identity operator.

\begin{table}[ht]
\centering
\begin{tabular}{|c|c|c|c|c|c|c|c|}
\hline
Measurements & MUBs & MUMs & SIC POVM & gSIC POVM & $(N,M)$-POVM \\
\hline
$N$ & $d+1$ & $d+1$ & $1$ & $1$ & $N$ \\
\hline
$M_\alpha$ & $d$ & $d$ & $d^2$ & $d^2$ & $M$ \\
\hline
$\gamma_\alpha$ & $\displaystyle\frac{1}{d+1}$ & $\displaystyle\frac{1}{d+1}$ & $1$ & $1$ & $\displaystyle\frac{1}{N}$ \\
\hline
$a_\alpha$ & $\displaystyle\frac{1}{d+1}$ & $\displaystyle\frac{1}{d+1}$ & $\displaystyle\frac{1}{d}$ & $\displaystyle\frac{1}{d}$ & $\displaystyle\frac{d}{NM}$ \\
\hline
$b_\alpha$ & $1$ & $\displaystyle\frac 1d<b\leq 1$ & $1$ & $\displaystyle\frac 1d<b\leq 1$ & $\displaystyle\frac 1d<b\leq \displaystyle\frac 1d \min\{d,M\}$ \\
\hline
$c_\alpha$ & $0$ & $\displaystyle\frac{1-b}{d-1}$ & $\displaystyle\frac{1}{d+1}$ & $\displaystyle\frac{d-b}{d^2-1}$ & $\displaystyle\frac{M-db}{d(M-1)}$ \\
\hline
$\mu$ & $\displaystyle\frac{1}{d(d+1)}$ & $\displaystyle\frac{1}{d(d+1)}$ & $\displaystyle\frac{1}{d^2}$ & $\displaystyle\frac{1}{d^2}$ & $\displaystyle\frac{1}{NM}$ \\
\hline
$S$ & $\displaystyle\frac{1}{(d+1)^2}$ & $\displaystyle\frac{db-1}{(d+1)(d^2-1)}$ & $\displaystyle\frac{1}{d(d+1)}$ & $\displaystyle\frac{db-1}{d(d^2-1)}$ & $\displaystyle\frac{d(db-1)}{NM(d^2-1)}$ \\
\hline
$\mathcal{C}_{\max}$ & $\displaystyle\frac{2}{(d+1)^2}$ & $\displaystyle\frac{b+1}{(d+1)^2}$ & $\displaystyle\frac{2}{d(d+1)}$ & $\displaystyle\frac{b+1}{d(d+1)}$ & $\displaystyle\frac{d(b+1)}{NM(d+1)}$ \\
\hline
\end{tabular}
\caption{Parameters characterizing conical 2-design GEAMs for popular classes of quantum measurements.}
\label{table1}
\end{table}

In what follows, we analyze a handful of measures constructed from conical 2-designs generalized equiangular measurements. We observe an interesting property. It turns out that all of the considered measures can be fully characterized by knowing only two constants: the symmetry constant $S$ and the maximal index of coincidence $\mathcal{C}_{\max}$. On a sidenote, we also manage to reproduce -- and rewrite in a simpler form -- many results known in the literature for MUMs and general SIC POVMs, while also generalizing these results to other classes of measurements.

\begin{Remark}
Alternatively, instead of knowing the values of $S$ and $\mathcal{C}_{\max}$, it is enough to know the constants that characterize the conical 2-design;
\begin{equation}
\kappa_+=\mathcal{C}_{\max}-S,\qquad \kappa_-=S.
\end{equation}
\end{Remark}

\section{Entropic uncertainty relations}

Entropy is a fundamental concept in information theory, statistical mechanics, and quantum physics, providing a quantitative measure of uncertainty. Given a probability distribution $p_j$, the Shannon entropy $H=-\sum_jp_j\log p_j$ measures the amount of information that it contains. In this section, we consider its two one-parameter generalizations: the R\'{e}nyi entropy \cite{Renyi}
\begin{equation}\label{R}
\mathcal{R}_\nu=\frac{1}{1-\nu}\log\left(\sum_jp_j^\nu\right)
\end{equation}
and the Tsallis entropy \cite{Tsallis}
\begin{equation}\label{T}
\mathcal{T}_\nu=-\sum_jp_j^\nu\log_\nu p_j,
\end{equation}
where the $\nu$-logarithm is defined via \cite{Rastegin5}
\begin{equation}
\log_\nu p_j=\frac{p_j^{1-\nu}-1}{1-\nu}.
\end{equation}
Specifically, for $\nu\to 1$, both formulas reproduce the standard Shannon entropy. The Tsallis entropy finds important applications e.g. in Bell inequalities \cite{Wajs}, quantum key distribution \cite{Kurzyk}, and quantum resource theory \cite{imaginarity}, whereas the R\'{e}nyi entropy is mainly used in quantification of quantum entanglement \cite{Franchini,Bertini,Mannai,ERE}.

Now, for the probability disctibution $p_{\alpha,k}=\Tr(\rho P_{\alpha,k})$ constructed from a GEAM $\mathcal{P}=\{P_{\alpha,k}\}$, we use eqs. (\ref{R}--\ref{T}) to construct
\begin{equation}
\mathcal{R}_\nu=\frac{1}{1-\nu}\log\left(\sum_{\alpha=1}^N
\sum_{k=1}^{M_\alpha}p_{\alpha,k}^\nu\right),
\end{equation}
\begin{equation}
\mathcal{T}_\nu=\frac{1}{1-\nu}\left(\sum_{\alpha=1}^N
\sum_{k=1}^{M_\alpha}p_{\alpha,k}^\nu-1\right).
\end{equation}
In what follows, we show that the generalized entropies are upper bounded by functions of the index of coincidence alone.

\begin{Theorem}
The Renyi $\nu$-entropies for the probability distribution $p_{\alpha,k}$ associated with a conical 2-design GEAM are upper bounded by
\begin{equation}\label{Rn}
\mathcal{R}_\nu\geq \frac{\nu}{2(1-\nu)}\log\mathcal{C}
\end{equation}
for any $\nu\geq 2$.
\end{Theorem}

\begin{proof}
Using the methodology from ref. \cite{SICMUB_entropic}, it is enough to prove that, for a probability distribution $p_j$,
\[\left(\sum_jp_j^\alpha\right)^{1/\alpha}\leq\left(\sum_jp_j^2\right)^{1/2}
=\sqrt{\mathcal{C}},\]
as long as $\nu\geq 2$. From this inequality, eq. (\ref{Rn}) immediately follows.
\end{proof}

\begin{Theorem}
The Tsallis $\nu$-entropies for the probability distribution $p_{\alpha,k}$ associated with a conical 2-design GEAM is upper bounded by
\begin{equation}\label{Hn}
\mathcal{T}_\nu\geq -\log_\nu\mathcal{C},
\end{equation}
for any $0<\nu\leq 2$.
\end{Theorem}

\begin{proof}
The proof is exactly as in Proposition 1 from ref. \cite{Rastegin5}. It uses Jensen's inequality \cite{Sanchez,Maassen,Dragomir}
\begin{equation}
\sum_jp_j\log_\nu p_j\leq \log\sum_jp_j^2
\end{equation}
in the range of $\nu$ where $\log_\nu$ is a convex function.
\end{proof}

In particular, the Shannon entropy ($\nu=1$) is bounded by
\begin{equation}
\mathcal{T}_1\geq-\log\mathcal{C}.
\end{equation}
Notably, less tight but easier to compute bounds for $\mathcal{R}_\nu$ and $\mathcal{T}_\nu$ are recovered if one replaces the index of coincidence $\mathcal{C}$ in eqs. (\ref{Rn}--\ref{Hn}) by its upper bound $\mathcal{C}_{\max}$. However, even without doing so, the upper bounds for both generalized entropies depend only on $S$ and $\mathcal{C}_{\max}$ due to
\begin{equation}
\mathcal{C}(\rho)=\mathcal{C}_{\max}+S(\Tr\rho^2-1).
\end{equation}
For the GEAMs that are not conical 2-designs, the entropies are bounded by a much more complicated formula, where the index of coincidence does not appear \cite{Rastegin_GEAM}.

Our results recover the entropy bounds for the following special subclasses of GEAMs: MUBs \cite{Rastegin5}, MUMs \cite{ChenFei,entropic_GSIC}, SIC POVMs \cite{Rastegin5}, general SIC POVMs \cite{entropic_GSIC}, and $(N,M)$-POVMs \cite{SICMUB_entropic}. Note that, when calculated for families of POVMs, the results in the original papers differ from ours by a constant factor due to rescaling.

\section{Brukner-Zeilinger invariants}

As an alternative to entropy, Brukner and Zeilinger \cite{BZI1,BZI2} introduce information measures for any finite dimensional quantum state $\rho$ over a complete set of non-degenerate observables $P_j$, where $p_j=\Tr(\rho P_j)$. The total information is defined as
\begin{equation}\label{I}
\mathcal{I}(\rho)=\sum_j\left(p_j-\frac 1d\right)^2=\Tr(\rho^2)-\frac 1d.
\end{equation}
The missing information is expressed via the total uncertainty
\begin{equation}\label{U}
\mathcal{U}(\rho)=\sum_jp_j(1-p_j)=1-\Tr(\rho^2).
\end{equation}
These are known as the Brukner-Zeilinger invariants, where both measures are trivially invariant with respect to any unitary transformations of the quantum state $\rho$. Under the action of bistochastic maps, these invariants can be only decreasing \cite{BZI5}.

An alternative approach to the Brukner-Zeilinger invariants is proposed by Luo \cite{BZI3}. The sum of variances for a complete orthogonal set of observables in the same quantum state results in the total variance
\begin{equation}\label{V}
\mathcal{V}(\rho)=\sum_j\left[\Tr(\rho P_j^2)-(\Tr\rho P_j)^2\right]=d-\Tr(\rho^2).
\end{equation}
It turns out that this is related to the Brukner-Zeilinger total information and total uncertainty by
\begin{equation}\label{IU}
\mathcal{I}(\rho)=\mathcal{V}_{\max}-\mathcal{V}(\rho),\qquad
\mathcal{U}(\rho)=\mathcal{V}(\rho)-\mathcal{V}_{\min}.
\end{equation}
Notably, the minimal $\mathcal{V}_{\min}=\mathcal{V}(\rho_\ast)$ and maximal $\mathcal{V}_{\max}=\mathcal{V}(P)$ variances are reached on the maximally mixed state $\rho_\ast=\mathbb{I}_d/d$ and a rank-1 projector $P$, respectively. Hence, $\mathcal{I}(\rho)$ and $\mathcal{U}(\rho)$ can be viewed as renormalized total variance \cite{BZI4}.

The total variance for the probability distribution $p_{\alpha,k}=\Tr(\rho P_{\alpha,k})$ associated with the GEAM $\mathcal{P}=\{P_{\alpha,k}\}$ is given by
\begin{equation}\label{VV}
\mathcal{V}(\rho)=\sum_{\alpha=1}^N\sum_{k=1}^{M_\alpha}\left[
\Tr(\rho P_{\alpha,k}^2)-(\Tr\rho P_{\alpha,k})^2\right]
=\Tr\left(\rho \sum_{\alpha=1}^N\sum_{k=1}^{M_\alpha}P_{\alpha,k}^2\right)
-\mathcal{C}(\rho),
\end{equation}
where $\mathcal{C}(\rho)$ is the index of coincidence. To simplify the trace formula, we need the following lemma.

\begin{Lemma}\label{suma}
If the GEAM $\mathcal{P}=\{P_{\alpha,k}\}$ forms a conical 2-design, then
\begin{equation}\label{P2}
\sum_{\alpha=1}^N\sum_{k=1}^{M_\alpha}P_{\alpha,k}^2=[\mathcal{C}_{\max}+(d-1)S]\mathbb{I}_d.
\end{equation}
\end{Lemma}

\begin{proof}
For any orthonormal operator basis $\{G_j;\,j=0,\ldots,d^2-1\}$, one has \cite{Werner_basis,Ohno}
\begin{equation}
\sum_{j=0}^{d^2-1}G_j^2=d\mathbb{I}_d.
\end{equation}
Taking the Hermitian orthonormal basis from eq. (\ref{Ga}), we find that
\begin{equation}\label{G2}
\sum_{\alpha=1}^N\sum_{k=1}^{M_\alpha-1}G_{\alpha,k}^2=\frac{d^2-1}{d}\mathbb{I}_d.
\end{equation}
Now, we use the relation between $G_{\alpha,k}$ and $P_{\alpha,k}$ from eq. (\ref{Gak}) to reformulate this property in terms of the measurement operators;
\begin{equation}
\sum_{\alpha=1}^N\sum_{k=1}^{M_\alpha}\frac{P_{\alpha,k}^2}{a_\alpha^2(b_\alpha-c_\alpha)}
=\left(\frac{d^2-1}{d}+\sum_{\alpha=1}^N\frac{\gamma_\alpha^2}{M_\alpha a_\alpha^2(b_\alpha-c_\alpha)}\right)
\mathbb{I}_d.
\end{equation}
This is a general property of any GEAM. Now, if $P_{\alpha,k}$ form a conical 2-design, then $a_\alpha^2(b_\alpha-c_\alpha)=S$, and the formula can be further simplified to
\begin{equation}
\sum_{\alpha=1}^N\sum_{k=1}^{M_\alpha}P_{\alpha,k}^2=[\mathcal{C}_{\max}+(d-1)S]\mathbb{I}_d,
\end{equation}
where
\begin{equation}
\mu=\sum_{\alpha=1}^N\frac{\gamma_\alpha^2}{M_\alpha},\qquad \mathcal{C}_{\max}=\frac{d-1}{d}S+\mu,
\end{equation}
with $\mathcal{C}_{\max}$ being the maximal index of coincidence.
\end{proof}

An immediate implication from eq. (\ref{P2}) is that $P_{\alpha,k}^2$ form a POVM after a proper rescaling. However, we know nothing about linear independency nor symmetry properties of such a measurement.

\begin{Theorem}\label{th1}
The Brukner-Zeilinger invariants for a conical 2-design GEAM read
\begin{enumerate}
\item the total variance
\begin{equation}
\mathcal{V}(\rho)=S[d-\Tr(\rho^2)],
\end{equation}
\item the invariant information
\begin{equation}
\mathcal{I}(\rho)=S\left[\Tr(\rho^2)-\frac 1d\right],
\end{equation}
\item the invariant uncertainty
\begin{equation}
\mathcal{U}(\rho)=S[1-\Tr(\rho^2)].
\end{equation}
\end{enumerate}
\end{Theorem}

\begin{proof}
Applying the results of Lemma \ref{suma} in eq. (\ref{VV}), we arrive at
\begin{equation}
\mathcal{V}(\rho)=\mathcal{C}_{\max}+(d-1)S-\mathcal{C}(\rho)
=S[d-\Tr(\rho^2)].
\end{equation}
It is straightforward to recover the upper and lower bounds of $\mathcal{V}(\rho)$,
\begin{equation}
\mathcal{V}_{\max}=\frac{d^2-1}{d}S,\qquad \mathcal{V}_{\min}=(d-1)S.
\end{equation}
Thus, $\mathcal{I}(\rho)$ and $\mathcal{U}(\rho)$ follow directly from eq. (\ref{IU}).
\end{proof}

Observe that the Brukner-Zeilinger invariants for GEAMs in Theorem \ref{th1} can be easily rescaled to equal their definitions from eqs. (\ref{I}--\ref{V}), which are valid for complete orthogonal sets of observables. Therefore, overcompleteness of GEAMs results only in multiplication by a constant factor. Our results reproduce the Brukner-Zeilinger invariants for MUBs \cite{BZI5}, MUMs \cite{BZI5,BZI6}, SIC POVMs \cite{BZI5}, general SIC POVMs \cite{BZI5,BZI6}, and $(N,M)$-POVMs \cite{SICMUB_BZ}.

\section{Coherence via skew information}

In quantum information, quantum coherence is an important resource with applications in quantum thermodynamics \cite{Gour_thermo,Michal_thermo,Lostaglio_thermo}, quantum networks \cite{Bibak,Mintert}, and quantum computing \cite{LiuFan}. Among a variety of coherence measures that includes e.g. relative entropy \cite{coh_entropy}, fidelity \cite{coh_fidelity}, trace distance \cite{Rana}, and $\ell_1$-norm \cite{l1norm}, it is skew information \cite{Girolami} that reveals a relation between quantum states and observables. The Wigner-Yanase skew information  $\mathcal{J}(\rho,A)=-\frac 12 \Tr([\rho^{1/2},A]^2)$ \cite{WY} for a quantum state $\rho$ is measured with respect to an observable $A$. It was later generalized to the Wigner-Yanase-Dyson skew information \cite{WYD,WYD2}
\begin{equation}\label{Jmu}
\mathcal{J}_\mu(\rho,A)=-\frac 12 \Tr([\rho^\mu,A][\rho^{1-\mu},A]),
\end{equation}
which is characterized via an additional parameter $0\leq\mu\leq 1$. Interestingly, there exists an intristic relation between $\mathcal{J}_\mu(\rho,A)$ and strong subadditivity of the von Neumann entropy \cite{LiebRuskai}. Finally, adding another free parameter results in the generalized Wigner-Yanase-Dyson skew information \cite{GWYD}
\begin{equation}\label{Jmunu}
\mathcal{J}_{\mu\nu}(\rho,A)=-\frac 12 \Tr([\rho^\mu,A][\rho^\nu,A]
\rho^{1-\mu-\nu}),
\end{equation}
where $\mu,\nu\geq 0$ and $\mu+\nu\leq 1$. This allows for greater flexibility for a broad range of states and measurement settings. Note that $\mathcal{J}_{\mu\nu}(\rho,A)$ is convex provided that $\mu+\nu=1$ \cite{WYD} or $\mu+\nu<1$, $2\mu+\nu\leq 1$, $\mu+2\nu\leq 1$ \cite{CaiLuo}.

Using the skew information, one introduces a coherence quantification based on quantum uncertainty. For the generalized Wigner-Yanase-Dyson skew information \cite{QUR2}, the quantum uncertainty is given by
\begin{equation}
\mathcal{Q}_{\mu\nu}(\rho,\mathcal{G})=\sum_{j=0}^{d^2-1}\mathcal{J}_{\mu\nu}(\rho,G_j),
\end{equation}
where $\mathcal{G}=\{G_j;\,j=1,\ldots,d^2\}$ is an orthonormal operator basis. Importantly, it is a convex function -- and therefore a well-defined coherence measure -- if and only if (i) $\mu+\nu=1$, or (ii) $\mu+\nu<1$, $2\mu+\nu\leq 1$, $\mu+2\nu\leq 1$. For $\mu+\nu=1$, one recovers the quantum uncertainty for the Wigner-Yanase-Dyson skew information \cite{Averaged_WYD},
\begin{equation}
\mathcal{Q}_\mu(\rho,\mathcal{G})=\sum_{j=0}^{d^2-1}\mathcal{J}_\mu(\rho,G_j),
\end{equation}
which has been shown to be independent on the choice of $\mathcal{G}$,
\begin{equation}\label{QG}
\mathcal{Q}_\mu(\rho,\mathcal{G})\equiv \mathcal{Q}_\mu(\rho)
=d-\Tr(\rho^\mu)\Tr(\rho^{1-\mu}).
\end{equation}
Analogically, for its generalization,
\begin{equation}\label{QG2}
\mathcal{Q}_{\mu\nu}(\rho,\mathcal{G})\equiv \mathcal{Q}_{\mu\nu}(\rho)
=\mathcal{Q}_\mu(\rho)+\mathcal{Q}_\nu(\rho)-\mathcal{Q}_{\mu+\nu}(\rho).
\end{equation}
However, the dependence on the mixed state $\rho$ persists. For the discussions on trade-offs between the coherence and mixedness of $\rho$, see refs. \cite{coh_mix,coh_mix2}.

Alternatively, the quantum uncertainty can be calculated for informationally (over)complete POVMs \cite{QUR}.

\begin{Theorem}
For a conical 2-design GEAM $\mathcal{P}$, the quantum coherence measures based on the Wigner-Yanase-Dyson skew information and its generalization read
\begin{equation}\label{QP}
\mathcal{Q}_\mu(\rho,\mathcal{P})=S \mathcal{Q}_\mu(\rho),
\end{equation}
\begin{equation}\label{Qmn}
\mathcal{Q}_{\mu\nu}(\rho,\mathcal{P})=\frac S2 [d+\Tr(\rho^{\mu+\nu}\rho^{1-\mu-\nu})-\Tr(\rho^\mu\rho^{1-\mu})
-\Tr(\rho^\nu\rho^{1-\nu})],
\end{equation}
respectively, where the entire dependence on $\mathcal{P}$ is through the measurement parameter $S$.
\end{Theorem}

\begin{proof}
From the definition of $\mathcal{J}_\mu(\rho,A)$ in eq. (\ref{Jmu}), it follows that
\begin{equation}
\mathcal{J}_\mu(\rho,A)=\Tr(\rho A^2)-\Tr(\rho^\mu A \rho^{1-\mu}A).
\end{equation}
Therefore, the quantum uncertainty for GEAMs
\begin{equation}
\mathcal{Q}_\mu(\rho,\mathcal{P})=\sum_{\alpha=1}^N\sum_{k=1}^{M_\alpha-1}
\mathcal{J}_\mu(\rho,P_{\alpha,k})
\end{equation}
simplifies to
\begin{equation}\label{QPmu}
\mathcal{Q}_\mu(\rho,\mathcal{P})=\Tr\left(\rho \sum_{\alpha=1}^N\sum_{k=1}^{M_\alpha}P_{\alpha,k}^2\right)
-\sum_{\alpha=1}^N\sum_{k=1}^{M_\alpha}\Tr(\rho^\mu P_{\alpha,k} \rho^{1-\mu}P_{\alpha,k})
=\mathcal{C}_{\max}+(d-1)S-\sum_{\alpha=1}^N\sum_{k=1}^{M_\alpha}\Tr(\rho^\mu P_{\alpha,k} \rho^{1-\mu}P_{\alpha,k}),
\end{equation}
where we used eq. (\ref{suma}) for the sum under the trace. Now, to find the formula for the second term, we need to calculate $\mathcal{Q}_\mu(\rho,\mathcal{G})$ for $\mathcal{G}$ defined by eqs. (\ref{Ga}) and (\ref{Gak}).

Let us start with $\mathcal{J}_\mu(\rho,\mathcal{G})$.
The first basis operator, $G_0=\mathbb{I}_d/\sqrt{d}$, and therefore $\mathcal{J}_\mu(\rho,G_0)=0$. For the remaining basis operators, one finds
\begin{equation*}
\begin{split}
\mathcal{J}_\mu(\rho,&G_{\alpha,k})=\frac{1}{S}
\Big[\Tr(\rho P_{\alpha,k}^2)-\Tr(\rho^\mu P_{\alpha,k}\rho^{1-\mu}P_{\alpha,k})\Big]
+\frac{1}{S(1+\sqrt{M_\alpha})^2}
\Big[\Tr(\rho P_{\alpha,M_\alpha}^2)-\Tr(\rho^\mu P_{\alpha,M_\alpha}\rho^{1-\mu}P_{\alpha,M_\alpha})\Big]\\&
-\frac{1}{SM_\alpha(1+\sqrt{M_\alpha})^2}
\Big[\Tr(\rho P_{\alpha,M_\alpha}P_{\alpha,k})+\Tr(\rho P_{\alpha,k}P_{\alpha,M_\alpha})-\Tr(\rho^\mu P_{\alpha,M_\alpha}\rho^{1-\mu}P_{\alpha,k})-\Tr(\rho^\mu P_{\alpha,k}\rho^{1-\mu}P_{\alpha,M_\alpha})\Big].
\end{split}
\end{equation*}
Now, taking a sum over all elements of $\mathcal{G}$, we obtain
\begin{equation}
\mathcal{Q}_\mu(\rho,\mathcal{G})=\sum_{\alpha=1}^N\sum_{k=1}^{M_\alpha-1}
\mathcal{J}_\mu(\rho,G_{\alpha,k})
=\frac 1S \left[\mathcal{C}_{\max}+(d-1)S
-\sum_{\alpha=1}^N\sum_{k=1}^{M_\alpha}\Tr(\rho^\mu P_{\alpha,k}\rho^{1-\mu}P_{\alpha,k})\right],
\end{equation}
where we used the results of Lemma \ref{suma} and the fact that $\sum_{k=1}^{M_\alpha}P_{\alpha,k}=\gamma_\alpha\mathbb{I}_d$. Finally, if we compare our results with eq. (\ref{QG}), it becomes clear that
\begin{equation}\label{suma2}
\sum_{\alpha=1}^N\sum_{k=1}^{M_\alpha}\Tr(\rho^\mu P_{\alpha,k}\rho^{1-\mu}P_{\alpha,k})=\mathcal{C}_{\max}+S\Big[\Tr(\rho^\mu)\Tr(\rho^{1-\mu})
-1\Big]=\mathcal{C}_{\max}+S\Big[d-1-\mathcal{Q}_\mu(\rho)\Big].
\end{equation}
Finally, going back to eq. (\ref{QPmu}), we input input the above formula and recover eq. (\ref{QP}).

As for proving the second formula, we once again start from the definition. Observe that $\mathcal{J}_{\mu\nu}(\rho,A)$ in eq. (\ref{Jmunu}) can be rewritten into
\begin{equation}\label{Imn}
\begin{split}
\mathcal{J}_{\mu\nu}(\rho,A)&=\frac 12 [\Tr(\rho A^2)+\Tr(\rho^{\mu+\nu}A\rho^{1-\mu-\nu}A)
-\Tr(\rho^\mu A\rho^{1-\mu}A)-\Tr(\rho^\nu A\rho^{1-\nu}A)]
\\&=\frac 12 [\mathcal{J}_{\mu}(\rho,A)+\mathcal{J}_{\nu}(\rho,A)-\mathcal{J}_{\mu+\nu}(\rho,A)].
\end{split}
\end{equation}
Hence, the quantum uncertainty for the generalized equiangular measurements reads
\begin{equation}
\mathcal{Q}_{\mu\nu}(\rho,\mathcal{P})=\frac 12 [\mathcal{Q}_{\mu}(\rho,\mathcal{P})+Q_{\nu}(\rho,\mathcal{P})-\mathcal{Q}_{\mu+\nu}(\rho,\mathcal{P})]
=\frac S2 [\mathcal{Q}_{\mu}(\rho)+\mathcal{Q}_{\nu}(\rho)-\mathcal{Q}_{\mu+\nu}(\rho)]
=\mathcal{Q}_{\mu\nu}(\rho).
\end{equation}
\end{proof}

Note that our results generalize the special cases of quantum uncertainties for MUBs, MUMs, SIC POVMs, general SIC POVMs, and $(N,M)$-POVMs \cite{SICMUB_App3,QUR2,GWYD_SICMUB}.

Aside from the aforementioned coherence measures, one considers the maximal coherence \cite{coh_WYD}
\begin{equation}
\mathcal{Q}_{\mu\nu}^{\max}(\rho)=\max_\Pi \mathcal{J}_{\mu\nu}(\rho,\Pi),
\end{equation}
where $\Pi$ is the von Neumann measurement, $\mu,\nu\geq 0$ and $\mu+\nu\leq 1$.
If $\mu+\nu<1$, then this equality holds if and only if $2\mu+\nu\leq 1$ and $\mu+2\nu\leq 1$. The explicit formula after maximization reads \cite{GWYD_SICMUB}
\begin{equation}
\mathcal{Q}_{\mu\nu}^{\max}(\rho)=\frac{1}{2d}\mathcal{Q}_{\mu\nu}(\rho).
\end{equation}
Now, observe that $\mathcal{Q}_{\mu\nu}^{\max}(\rho)=q_1\mathcal{Q}_{\mu\nu}(\rho,\mathcal{P}_1)+q_2\mathcal{Q}_{\mu\nu}(\rho,\mathcal{P}_2)$ for the conical 2-design GEAMs $\mathcal{P}_1$ and $\mathcal{P}_2$ characterized via $S_1$ and $S_2$, respectively, with $q_1S_1+q_2S_2=1/d$. One example was given in ref. \cite{GWYD_SICMUB}, where $\mathcal{P}_1$ is a SIC POVM with $S_1=\frac{1}{d(d+1)}$, $\mathcal{P}_2$ is a maximal set of $N=d+1$ MUBs with $\gamma_\alpha=1/N$, and $q_1=1$, $q_2=d+1$ are the numbers of respective POVMs, respectively.

\section{Schmidt number criterion}

In the theory of quantum entanglement, the Schmidt number is an entanglement quantifier for bipartite states \cite{Terhal}. It introduces a hierarchy of entangled states, changing from $1$ (separable states) up to $d$ \cite{Sperling,TOPICAL}. The Schmidt number criteria have been obtained using e.g. fidelity \cite{Terhal}, cross norm \cite{Johnston4}, Bloch decomposition \cite{Klockl}, covariance matrix \cite{Liu4,Liu5}, and witnesses \cite{Hulpke,Sanpera,Wyderka,Shi}.

In this section, we use the correlation matrix constructed from conical 2-design generalized equiangular measurements to derive a Schmidt number criterion. Following the method in ref. \cite{lower_conc}, we take a set $\omega_{\alpha,k}$ of orthonormal vectors in $\mathbb{C}^{d^2+N-1}$ and define the correlation operator,
\begin{equation}\label{bro}
\mathcal{B}(\rho)=\sum_{\alpha,\beta=1}^N\sum_{k=1}^{M_\alpha}\sum_{\ell=1}^{M_\beta}
\Tr[\rho(P_{\alpha,k}\otimes P_{\beta,\ell})]|\omega_{\alpha,k}\>\<\omega_{\beta,\ell}|.
\end{equation}
The elements of $\mathcal{B}(\rho)$ in the basis of $\omega_{\alpha,k}$ form the correlation matrix $\mathcal{B}_{\alpha,k;\beta,\ell}=\Tr[\rho(P_{\alpha,k}\otimes P_{\beta,\ell})]$. The entanglement quantifier is derived with respect to the trace norm $\|\mathcal{B}(\rho)\|_{\tr}=\Tr\sqrt{\mathcal{B}(\rho)^\dagger \mathcal{B}(\rho)}$ of the correlation operator $\mathcal{B}(\rho)$. First, however, we need the following lemma.

\begin{Lemma}\label{Lemma1}
For any pure bipartite state $|\psi\>$, eq. (\ref{bro}) produces
\begin{equation}\label{Bpsi}
\|\mathcal{B}(|\psi\>\<\psi|)\|_{\tr}=\mathcal{C}_{\max}+2S\sum_{j<k}\lambda_j\lambda_k,
\end{equation}
where $\lambda_j$ such that $\lambda_j\geq 0$ and $\sum_{j=1}^r\lambda_j^2=1$ are defined via the Schmidt decomposition
\begin{equation}\label{psi}
|\psi\>=\sum_{j=1}^r\lambda_j|e_j\>\otimes|f_j\>,
\end{equation}
with the Schmidt number $r$ and two orthonormal bases $\{e_j\}$, $\{f_j\}$ in $\mathbb{C}^d$.
\end{Lemma}

\begin{proof}
First, we use GEAMs to define two families of vectors:
\begin{equation}
|u_{ij}\>=\sum_{\alpha=1}^N\sum_{k=1}^{M_\alpha}\<e_j|P_{\alpha,k}|e_i\>
|\omega_{\alpha,k}\>,
\end{equation}
\begin{equation}
|v_{ij}\>=\sum_{\alpha=1}^N\sum_{k=1}^{M_\alpha}\<f_j|P_{\alpha,k}|f_i\>
|\omega_{\alpha,k}\>,
\end{equation}
where $i,j=1,\ldots,d$. Now, for the vector $|\psi\>$ in eq. (\ref{psi}), we calculate
\begin{equation}
\mathcal{B}(|\psi\>\<\psi|)=\sum_{i,j=1}^r\lambda_i\lambda_j|u_{ij}\>
\<\overline{v}_{ij}|,
\end{equation}
where $\overline{x}$ is a complex conjugation of $x$. Recall that there always exists a unitary operator $U$ such that $|\overline{v}_{ij}\>=U|u_{ij}\>$. Therefore, from the unitary invariance of the trace norm $\|X\|_{\tr}=\|UXV\|_{\tr}$ for any unitary $U$ and $V$, we find that
\begin{equation}\label{BB}
\|\mathcal{B}(|\psi\>\<\psi|)\|_{\tr}=\left\|\sum_{i,j=1}^r\lambda_i\lambda_j|u_{ij}\>
\<u_{ij}|\right\|_{\tr}=\sum_{i,j=1}^r\lambda_i\lambda_j\<u_{ij}|u_{ij}\>.
\end{equation}
Next, the orthonormality of $|\omega_{\alpha,k}\>$ and the property of conical 2-designs in eq. (\ref{con}) allow us to compute
\begin{equation}
\begin{split}
\<u_{ij}|u_{ij}\>&=\sum_{\alpha,\beta=1}^N\sum_{k=1}^{M_\alpha}\sum_{\ell=1}^{M_\beta}
\<\omega_{\beta,\ell}|\overline{\<e_j|P_{\beta,\ell}|e_i\>}\<e_j|P_{\alpha,k}|e_i\>
|\omega_{\alpha,k}\>
\\&=\sum_{\alpha=1}^N\sum_{k=1}^{M_\alpha}\<e_i|P_{\alpha,k}|e_j\>\<e_j|P_{\alpha,k}|e_i\>
=\<e_i\otimes e_j|\sum_{\alpha=1}^N\sum_{k=1}^{M_\alpha}P_{\alpha,k}\otimes P_{\alpha,k}|e_j\otimes e_i\>
\\&=\<e_i\otimes e_j|(\kappa_+\mathbb{I}_d\otimes\mathbb{I}_d
+\kappa_-\mathbb{F}_d)|e_j\otimes e_i\>
=\kappa_+\delta_{ij}+\kappa_-.
\end{split}
\end{equation}
Hence, eq. (\ref{BB}) simplifies to
\begin{equation}\label{BB2}
\|\mathcal{B}(|\psi\>\<\psi|)\|_{\tr}=\sum_{i,j=1}^r\lambda_i\lambda_j
(\kappa_+\delta_{ij}+\kappa_-)=\kappa_+\sum_{i=1}^r\lambda_i^2
+\kappa_-\sum_{i,j=1}^r\lambda_i\lambda_j
=\kappa_++\kappa_-+2\kappa_-\sum_{i<j}\lambda_i\lambda_j,
\end{equation}
which is equivalent to eq. (\ref{Bpsi}).
\end{proof}

Now, it is straightforward to establish the Schmidt number criterion for a bipartite mixed state $\rho$.

\begin{Theorem}
If the Schmidt number of a bipartite state $\rho$ is at most $r$, then
\begin{equation}
\|\mathcal{B}(\rho)\|_{\tr}\leq \mathcal{C}_{\max}+(r-1)S.
\end{equation}
\end{Theorem}

\begin{proof}
The trace norm is convex, so it is enough to prove this theorem for pure states of rank $r$. From eq. (\ref{Bpsi}) in Lemma \ref{Lemma1} and $\sum_{j=1}^r\lambda_j\leq\sqrt{r}$ \cite{Terhal}, one has
\begin{equation}
\|\mathcal{B}(|\psi\>\<\psi|)\|_{\tr}=\kappa_++\kappa_-\sum_{i,j=1}^r\lambda_i\lambda_j
\leq \kappa_++r\kappa_-=\mathcal{C}_{\max}+(r-1)S.
\end{equation}
\end{proof}

Our results generalize the Schmidt number criterion obtained from SIC POVMs and MUBs \cite{Morelli}, general SIC POVMs \cite{Wangs}, and $(N,M)$-POVMs \cite{Schmidt_NM}. Moreover, a simple sufficient separability condition immediately follows.

\begin{Corollary}\label{cor}
If $\rho$ is separable, then $r=1$ and
\begin{equation}
\|\mathcal{B}(\rho)\|_{\tr}\leq \mathcal{C}_{\max},
\end{equation}
which reproduces some known separability criteria: ESIC (for SIC POVMs \cite{ESIC} and general SIC POVMs \cite{gESIC}) and generalized ESIC (for $(N,M)$-POVMs \cite{SIC-MUB} and generalized symmetric measurements \cite{SIC-MUB_general}).
\end{Corollary}

\section{Concurrence lower bound}

An alternative measure of quantum entanglement is concurrence \cite{Wooters1,Wooters2}. On bipartite pure states $\psi$, it is defined by \cite{Rungta,Bhaskara}
\begin{equation}
\mathcal{N}(\psi)=\sqrt{2[1-\Tr(\rho_1^2)]},\qquad \rho_1=\Tr_2(|\psi\>\<\psi|),
\end{equation}
where $\Tr_2 X$ is a partial trace of $X$ with respect to the second subsystem. The mixed states extension is typically established through the convex roof
\cite{Demkowicz}
\begin{equation}
\mathcal{N}(\rho)=\min_{\{p_j,\psi_j\}}\sum_jp_j\mathcal{N}(\psi_j).
\end{equation}
The minimum is taken over all the pure state decompositions $\rho=\sum_jp_j|\psi_j\>\<\psi_j|$ with probability distributions $p_j$. This optimization requirement makes calculating $\mathcal{N}(\rho)$ analytically an NP-hard problem. Therefore, of particular interest are derivations of its lower bound. Here, we follow the methodology and generalize the results of ref. \cite{lower_conc} derived for $(N,M)$-POVMs.

\begin{Theorem}
The concurrence of a mixed bipartite state $\rho$ is lower bounded by
\begin{equation}
\mathcal{N}_{\min}(\rho)=
\eta\Big[\|\mathcal{B}(\rho)\|_{\tr}-\xi\Big],
\end{equation}
where $\mathcal{B}(\rho)$ is the correlation operator from eq. (\ref{Bpsi}) and
\begin{equation}
\eta=\frac{1}{S}\sqrt{\frac{2}{d(d-1)}},\qquad \xi=\mathcal{C}_{\max}.
\end{equation}
\end{Theorem}

\begin{proof}
Assume that $\rho=\sum_jp_j|\psi_j\>\<\psi_j|$ is the optimal pure state decomposition of $\rho$. Then, its concurrence reads $\mathcal{N}(\rho)=\sum_jp_j\mathcal{N}(\psi_j)$. Now, it is enough to show that, for any pure state $\psi$,
\begin{equation}\label{conc}
\mathcal{N}(\psi)\geq\mathcal{N}_{\min}(\psi)=\eta\Big[\|\mathcal{B}(|\psi\>\<\psi|)\|_{\tr}-\xi\Big]
\end{equation}
for some constants $\eta$ and $\xi$ that have to be derived. Note that, for the mixed state $\rho$, eq. (\ref{conc}) implies
\begin{equation}
\mathcal{N}(\rho)=\sum_jp_j\mathcal{N}(\psi_j)\geq 
\eta\left[\sum_jp_j\|\mathcal{B}(|\psi_j\>\<\psi_j|)\|_{\tr}-\xi\right]
\geq \eta\left[\Big\|\sum_jp_j\mathcal{B}(|\psi_j\>\<\psi_j|)\Big\|_{\tr}-\xi\right]
=\eta\Big[\|\mathcal{B}(\rho)\|_{\tr}-\xi\Big],
\end{equation}
where we used the triangle inequality
\begin{equation}
\|\mathcal{B}(\rho)\|_{\tr}\leq \sum_jp_j\|\mathcal{B}(|\psi_j\>\<\psi_j|)\|_{\tr}.
\end{equation}
Next, from eq. (\ref{Bpsi}) in Lemma \ref{Lemma1}, we take the formula for $\|\mathcal{B}(|\psi_j\>\<\psi_j|)\|_\tr$ and substitute it into eq. (\ref{conc}). This results in
\begin{equation}
\mathcal{N}_{\min}(\psi)=\eta\left[\mathcal{C}_{\max}+2S\sum_{j<k}\lambda_j\lambda_k-\xi\right].
\end{equation}
Finally, from ref. \cite{Albeverio}, we know that
\begin{equation}
\mathcal{N}_{\min}(\psi)=\frac{2\sqrt{2}}{\sqrt{d(d-1)}}\sum_{j<k}\lambda_j\lambda_k,
\end{equation}
and hence we recover the values for $\eta$ and $\xi$.
\end{proof}

If $\rho$ is separable, then $\mathcal{N}(\rho)=0$, which leads back to the results of Corollary \ref{cor}.

\section{Conclusions}

In this paper, we analyzed the applicational properties of conical 2-designs constructed from generalized equiangular measurements. Our main result is proving that many quantum measures often considered in the literature, when constructed from these measurement operators, depend only on two constant parameters. This was shown on the example of entropic uncertainty relations, the Brukner-Zeilinger invariants, quantum coherence, quantum concurrence, and the Schmidt-number criterion. Our results generalize the formulas so far obtained for subclasses of GEAMs, while also offering a concise, simple form for each measure. The number and variety of quantities fully characterized by two measurement parameters indicates on a more general property of conical 2-designs. We conjecture that there exist strictly defined conditions under which measures constructed from GEAMs are only functions of $S$ and $\mathcal{C}_{\max}$ (alternatively, $\kappa_+$ and $\kappa_-$). This is a very interesting open problems for future studies.

\section{Acknowledgements}

This research was funded in whole or in part by the National Science Centre, Poland, Grant number 2021/43/D/ST2/00102. For the purpose of Open Access, the author has applied a CC-BY public copyright licence to any Author Accepted Manuscript (AAM) version arising from this submission.

\bibliography{C:/Users/cyndaquilka/OneDrive/Fizyka/bibliography}
\bibliographystyle{C:/Users/cyndaquilka/OneDrive/Fizyka/beztytulow2}

%
%

\end{document}